\documentclass[latexsym,amsfonts,lineno,preprint,pre,floats,aps,amsmath,amssymb,12]{article}
\usepackage[utf8]{inputenc}
\usepackage[T1]{fontenc} 
\usepackage[english]{babel} 
\usepackage[hmarginratio=1:1,top=32mm,columnsep=20pt]{geometry} 
\usepackage{graphicx}
\usepackage{bm}
\usepackage{abstract} 
	
\usepackage{titlesec} 
\titleformat{\section}[block]{\large\scshape\centering}{\thesection.}{1em}{} 
\titleformat{\subsection}[block]{\large}{\thesubsection.}{1em}{} 
\usepackage{titling} 

\usepackage{charter} 
\usepackage{hyperref} 
\usepackage{amssymb}
\usepackage{amsmath}
\usepackage{breqn}
\usepackage{amsfonts}
\usepackage{adjustbox}
\usepackage{graphicx}
\usepackage{caption}
\usepackage{subfig}
\usepackage{float}
\usepackage{tabularx}
\usepackage{fancyhdr} 
\usepackage{amsthm}
\usepackage{mathtools}
\usepackage{amssymb}
 
\usepackage{cases}

\newtheorem{lem}{Lemme}[section]
\newtheorem{thm}{Th\'eor\`eme}[section]

\providecommand{\keywords}[1]
{
	\small	
	\textbf{\textit{Keywords---}} #1
}
\title{Study of COVID-19 anti-pandemic strategies by using optimal control}
\author{
	Fulgence Mansal $^{a,b}$, Mouhamadou A.M.T. Bald\'e\thanks{Support of the Non Linear Analysis, Geometry and Applications (NLAGA) Project}\ \ $^{b,c}$ and Alpha O. Bah $^d$\\
	$^a$ Universit\'e Catholique de l'Afrique de l'Ouest/ UUZ\\
	\href{mailto:fulgence.mansal1@gmail.com}{fulgence.mansal@ucad.edu.sn}\\
	$^b$ Laboratory of Mathematics of Decision and Numerical Analysis (LMDAN).\\	
	$^c$ Department of Mathematics of Decision(DMD)-FASEG.\\
	University of Cheikh Anta Diop.  \\ 
	BP 45087, 10700. Dakar, Senegal.\\
	\href{mailto:mouhamadouamt.balde@ucad.edu.sn}{mouhamadouamt.balde@ucad.edu.sn}\\
	$^d$ Laboratory of Applied Mathematics (LMA)-FST.\\
	 University of Cheikh Anta Diop.\\
	\href{mailto:alphaoumarbah6@yahoo.fr}{alphaoumarbah6@yahoo.fr}\\
}	
\date{}
\renewcommand{\maketitlehookd}{%
	\begin{abstract}
In this study, we present a new epidemiological model, with contamination from confirmed and unreported. We also compute equilibria and study their stability without intervention strategies. 
Optimal control theory has proven to be a successful tool in understanding ways to curtail the spread of infectious diseases by devising the optimal disease intervention strategies. We investigate the impact of distancing, case finding, and case holding controls while at the same time, we minimize the number of infected and dead individuals. The method consists of minimizing the cost functional related to infectious, death, and controls through some strategies to reduce the spread of the COVID19 epidemic.   
	\end{abstract}
	\keywords{Optimal control strategy, Optimal control theory, Covid 19, Stability, Parameters esimates, Fitting COVID19 data}
}

\begin{document}
\maketitle

 \section{Introduction}
\noindent The COVID-19 pandemic has continued to evolve for more than six months around the world. Many countries have applied containment measures, then deconfinement. Currently, because of the resurgence of cases, some of these countries are proceeding with re-containment measures. \\

\noindent The health and social distancing measures are not always respected by the populations. Among the confirmed cases, there are caregivers. That shows a security flaw in the quarantine procedures. There are many undetected cases in the people who favor the evolution of the pandemic. \\
Hopes are on the discovery of a vaccine. But in the meantime, it is useful to come up with strategies that allow us to manage the pandemic better.

\noindent Several recent works have used SIR / SEIR models and other types of nonlinear differential equations(\cite{Balde:2020}, \cite{Balde2:2020}, \cite{Magal}) to understand the evolution of the pandemic but also to predict its subsequent evolution. Other techniques are also used such as machine learning, stochastic (\cite{Balde:2020}, \cite{Balde2:2020}, \cite{Ndiayeetal:2020}, \cite{Ndiayeetal2:2020}), etc. There is some work dedicated to the application of optimal control to the pandemic.
Many authors has used optimal control to study some diseases like HIV \cite{Hem}, \cite{Urszula}. \\

\noindent We analyze an epidemiological differential equation model with the identification of its parameters and initial values, based upon reported case data from public health sources. 
The objective of this work is to develop control strategies to stem the evolution of the pandemic.\\

\noindent The paper is organized as follows. In section \ref{sec:model}, we present the model. In section \ref{sec:analysis}, a mathematical analysis of the model is performed. Then, in section \ref{sec:optcont}, we introduce an optimal control problem to study. Thus in section \ref{sec:numsim}, we show numerical results of the optimal control problem. We discuss the results in section \ref{sec:discu}. We explain the methods we use in this work in section \ref{app}. Finally, we give conclusions and perspectives in section \ref{cp}.    

\section{Model formulation}
\label{sec:model}
%
%
%
%
\begin{figure}[H]
    \begin{center}
   \includegraphics[width=0.8\linewidth]{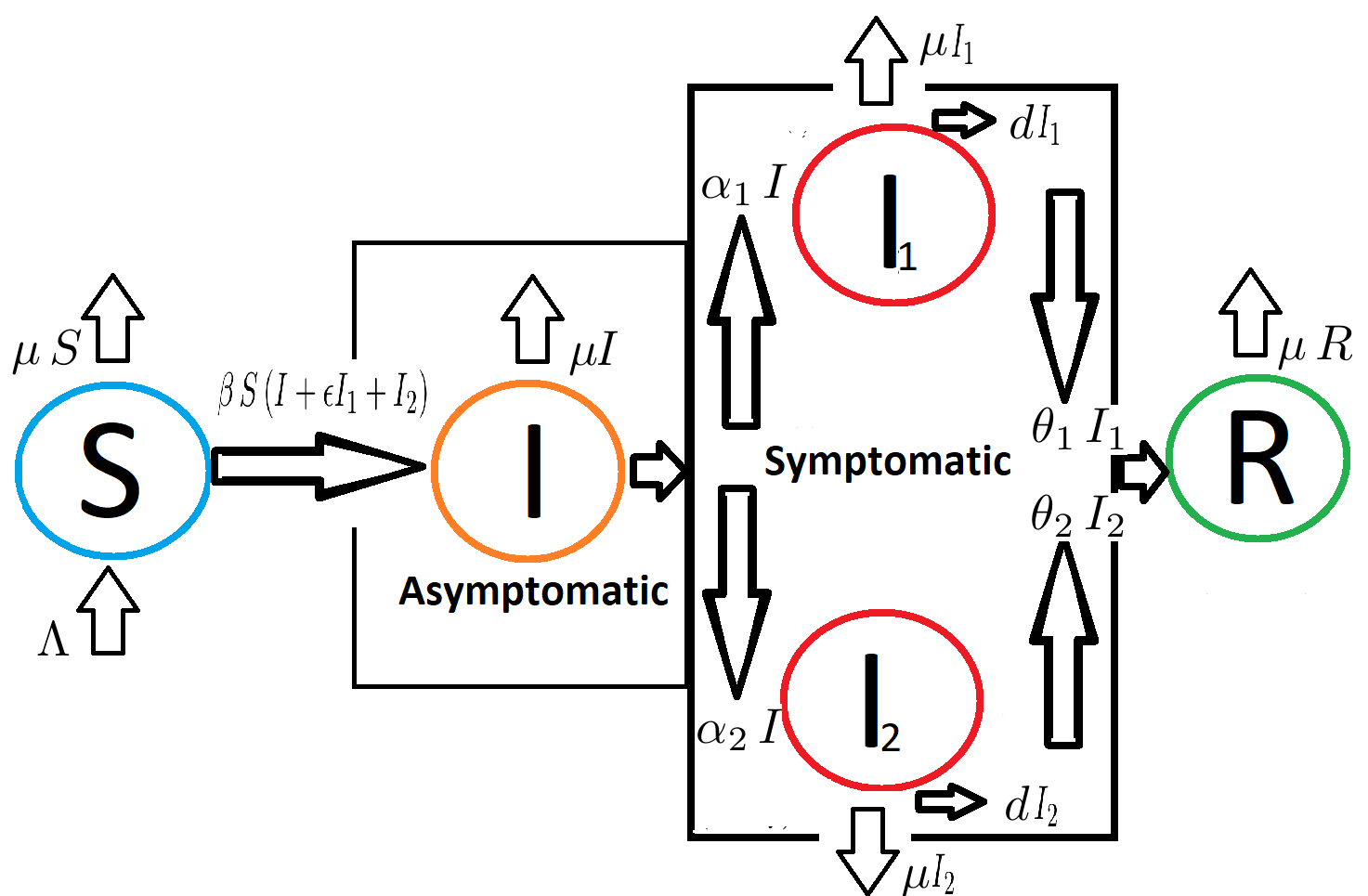}
  \caption{Compartments and flow chart of the model.}
    \label{fig:Figp.1}
    \end{center}
    \end{figure}

We consider the following differential equation model:
\begin{equation}
\label{model}
\left\{\begin{array}{lcl}
\dot S &=& \Lambda -\beta \,S\,(I+\epsilon I_{1}+I_{2})-\mu \,S,\\
\\
\dot I &=&\beta \,S\,(I+\epsilon I_{1}+I_{2})-(\mu+\alpha_1+\alpha_2)\,I,    \\
\\
\dot I_1&=& \alpha_1 \,I - (\mu+d+\theta_1)\,I_1,\\
\\
\dot I_2 &=&  \alpha_2 \,I - (\mu+d+\theta_2)\,I_2, \\
\\
\dot R &=&\theta_1\,I_1+ \theta_2\,I_2-\mu\,R.
\end{array}\right.
\end{equation}

\begin{table}[htp]
\caption{The model variables }
\begin{center}
\begin{tabular}{|c|c|c|c|c|c|c|c|}
\begin{tabularx}{15cm}{|c|X|}
\hline
Variable & Explanations for different classes  \\
\hline
$ S(t)$& Number of susceptible population at time $t$ \\
\hline
$I(t)$& Number of infected population at time $t$ (i.e. asymptomatic infectious)\\
\hline
$I_{1} (t)$ & Number of infected reported population  at time $t$  (i.e. symptomatic infectious with sever symptoms)\\
\hline
$I_{2}(t)$& Number of infected unreported population  at time $t$ (i.e., symptomatic infectious with mild symptoms) \\
\hline
$R(t)$& Number of recovered adults satisfying undetectable criteria at time $t$ \\
\hline
\end{tabularx}
\end{tabular}
\end{center}
\label{modelvar}
\end{table}%

\begin{table}[htp]
\caption{Parameters model formulation and their description}
\begin{center}
\begin{tabular}{|c|c|c|c|c|c|c|c|}
\hline
Parameter & Description  \\
\hline
$\alpha_1$ & Rate at which asymptomatic infectious become reported symptomatic \\
\hline
$\alpha_2$ & Rate at which asymptomatic infectious become unreported symptomatic \\
\hline
$\beta$ & Rate of transmission \\
\hline
$\mu$ & Natural death rate of the population \\
\hline
$ \Lambda $ &  Recruitment rate \\ 
\hline
$\theta_1$& Rate of recovery   from  reported population\\
\hline
$\theta_2$& Rate of recovery  unreported population \\
\hline
$d$ & Death rate of infected population due to COVID-19 coronavirus   \\
\hline

\end{tabular}
\end{center}
\label{modelparam}
\end{table}%
\noindent The system is supplemented by initial conditions
\begin{equation}
 S(t_0)=  S_0 > 0 \, , I(t_0)=  I_0 > 0 \, , I_1 (t_0)=  I_{10} = 0 \,\mbox{ and }  I_2 (t_0)=  I_{20} \ge  0,
 \end{equation}
with $t_{0}$ the starting time of the epidemic. Figure \ref{fig:Figp.1} depicts a flow diagram of the model. \\
In this model, the confirmed are automatically quarantined. Even in general, there may have some stages before confirmed individuals become quarantined. In this work, we consider these two compartments as one. Security failing during the quarantine or the process of quarantine can expose susceptible people to contamination. That is modeled by the term $\beta S \epsilon I_{1}$. Then a proportion $\epsilon$ of confirmed $I_{1}$ can interact with susceptible.    
\section{Mathematical analysis}
\label{sec:analysis}
One of the most critical concerns about any infectious disease is its ability to invade a population. The basic reproduction number,  $\mathcal R_0$ is a measure of the potential for disease spread in a population.   It represents the average number of secondary cases generated by an infected individual introduced into a susceptible population with no immunity to the disease in the absence of interventions to control the infection. If $ \mathcal R_0 < 1$, then on average, an infected individual produces less than one newly infected individual throughout his infection period. In this case, the infection may die out in the long run. Reversely, if $ \mathcal R_0 > 1$, each infected individual produces, on average more than one new infection. Hence the disease will be able to spread in a population. A significant value of $\mathcal R_0$ may indicate the possibility of a major epidemic. Using the next-generation operator technique described by  (\cite{Diekmann})  and subsequently analyzed by  (\cite{Van}),  we obtained the basic reproduction number. 
\subsection{Well–posedness of the model}
In this section, we prove that the system (\ref{model}) is epidemiologically meaningful. In other words, solutions of system (\ref{model}) with positive initial data remain positive for all time $t>0.$
Now, adding all equations in the differential system (\ref{model}) gives
\begin{equation*}
\dot N= \Lambda- \mu N- d_1 I_1-d_2 I_2\leq \Lambda- \mu N.
\end{equation*}
It then follows that, $\lim_{t\mapsto \infty}N(t)= \dfrac{\Lambda}{\mu}$ which implies that the trajectories of system (\ref{model})
are bounded. On the other hand, solving the differential inequality
\begin{equation}
\dot N \leq \Lambda- \mu N,
\end{equation}
so that, 
\begin{equation*}
0\leq N(t)\leq \dfrac{\Lambda}{\mu} +(N(0)-\dfrac{\Lambda}{\mu})e^{-\mu t}.
\end{equation*}
Thus, at $t \mapsto \infty, 0\leq N(t)\leq \dfrac{\Lambda}{\mu}.$ Therefore, all feasible solutions of model system (\ref{model}) 
enter the region:
\begin{equation}
\mathcal D=\big\{(S, I, I_1, I_2, R) \in \mathbb R_+^5 \big\}
\end{equation}
which is a positively invariant set of system (\ref{model}) .
Furthermore, the model (\ref{model}) is well-posed epidemiologically and
we will consider dynamic behavior of model (\ref{model}) on $ \mathcal D$.
%

\subsection{Equilibrium point}
To obtain the disease-free equilibrium, $I(t), I_1(t), I_2(t)$ and the right-hand-side of system
(\ref{model}) are set to zero. Then, the disease-free equilibrium will be given by
\begin{equation}
 E=\big(S_0, 0, 0, 0,0\big)=\big(\dfrac{\Lambda}{\mu},0,0,0,0\big)
\end{equation}
\subsection{Rate reproduction number}
By using the next-generation operator method on the system (\ref{model}), we establish the linear stability of $E_0$. Using the notation in \cite{Van}, the matrices F and V, for the new infection terms and the remaining transfer terms respectively, are given by (noting that $S_0=\dfrac{\Lambda}{\mu}$ at the DFE $E_0$)
$$\begin{array}{ll}F=\left(\begin{array}{ccc}
\beta S_0 & \epsilon\beta S_0 &  \beta S_0\\
\alpha_{1}& 0 & 0\\
\alpha_{2}& 0 & 0
\end{array}\right),\qquad V=\left(\begin{array}{cccc}
\mu+\alpha_1+\alpha_2 &0 & 0\\
0 & \mu+d+\theta_1  &0\\
0 &0 & \mu+d+\theta_2
\end{array}\right).
\end{array}$$
Thus,
\begin{equation}
 \mathcal R_0=\rho(FV^{-1})= \dfrac{\beta S_0}{\mu+\alpha_1+\alpha_2}= \dfrac{\beta \Lambda}{\mu(\mu+\alpha_1+\alpha_2)}.
\end{equation}
The following result follows from Theorem 2 of \cite{Van}
\begin{lem}
 The DFE of the Covid 19-only model (\ref{model}),  is locally asymptotically stable (LAS) if
 $ \mathcal R_0 < 1$, and unstable if $ \mathcal R_0 > 1$.
\end{lem}

The threshold quantity of $\mathcal R_0$ is the reproduction number for COVID-19. It measures the average number of new Covid-19 infections generated by a single COVID-19 infected individual in a population where a certain fraction of infected individuals is treated.
\subsection{Global stability of the disease-free equilibrium}
We now turn to the global stability of the disease-free equilibrium $E_0$. We prove that the disease-free equilibrium $E_0$ is globally
asymptotically stable under a certain threshold condition. To this
aim, we use a result obtained by Kamgang and Sallet  \cite{kamg}.

Let $x_1=(S, R)$ and $x_2=(I_2, I_1, I)$. We express the sub-system
\begin{equation}
 \dot x_1=A_1(x_1,0).(x-x_1^*),  \quad \text{as}
\end{equation}

\begin{equation*}
\left\{\begin{array}{lll}
\dot S &=& \Lambda -\mu \,S,\\
\dot R &=&-\mu\,R.
\end{array}\right.
\end{equation*}
It is a linear system which  is globally asymptotically stable at the equilibrium $E_0$, corresponding to the DFE
where the hypotheses $\mathbb H_1$ and $\mathbb H_2$ in \cite{kamg} are satisfied.

The matrix $\mathbb A_2(x)$ is given by
$$\begin{array}{ll}\mathbb A_2(x)=\left(\begin{array}{ccc}
-(\mu+d_2+\theta_2)  &0 & \alpha_2 \\
\\
0& - (\mu+d_1+\theta_1)  &\alpha_1 \\
\\
0&0&\beta S_0- (\mu+\alpha_1+\alpha_2) 
\end{array}\right).
\end{array}$$
The eigenvalues of the sub-matrix:
$$\begin{array}{ll}\mathbb J_0=\left(\begin{array}{cc}
 - (\mu+d_1+\theta_1)  &\alpha_1 \\
\\
0&\beta S_0- (\mu+\alpha_1+\alpha_2) 
\end{array}\right).
\end{array}$$
Since $J_0$ is a matrix of dimension 2, necessaries conditions for $J_0$ to be stable is $tr(J_0)<1$ and $det (J_0)>0$.
Note that $tr(J_0)<0$ gives  $\beta S_0< \mu+\alpha_1+\alpha_2.$ Also, the condition $det (J_0)>0$ gives
\begin{equation}\label{alfa}
 \beta S_0\leq \mu+\alpha_1+\alpha_2.
\end{equation}
Note that the inequality (\ref{alfa}) corresponds to $\mathcal R_0 \leq 1$. This achieves the proof.

\noindent We have the following result about the stability of the disease-free equilibrium.
\begin{thm}
The disease-free equilibrium of system (\ref{model}) is globally asymptotically stable in $D$ whenever $\mathcal R_0 \leq 1$. This
implies the global asymptotic stability of the disease-free equilibrium on the nonnegative orthant $\mathbb R_5$, i.e., 
the disease naturally dies out.
\end{thm}

\begin{proof}

 We consider the Lyapunov function defined by $ V(S, I, I_1, I_{2})=I$ .   So, we have
 
 \begin{eqnarray*}
  \dot V &=& \dot I   \\ 
   &=& \beta SI - (\mu+ \alpha_1+\alpha_2)I   \\
   &=& I\big[\beta S - (\mu+ \alpha_1+\alpha_2) \big]  \\
   &=& I(\mu+ \alpha_1+\alpha_2)\big[ \dfrac{\mathcal R_0 S}{S_0} -1 \big]  \\
   &\leq & 0 .
 \end{eqnarray*}

\noindent Moreover $\dot V = 0$ if $I=0$ or $S=S_0$ and $\mathcal R_0=1$. Since we are in a positively invariant compact,
by LaSalle's invariance principle \cite{LaSalle}, the DFE is globally asymtotically stable in $\mathcal D$.

\end{proof}

\section{Optimal control in the epidemic model}
\label{sec:optcont}
Optimal control problems have generated a lot of interest from researchers all over the world. For instance (Imanov, 2011, see \cite{Imanov}) examined the application of the method of similar solutions in solving time optimal control problems with state constraints. Similarly, various techniques have been applied to study optimal control problems related to dynamical systems. However, we considered the aspect of optimal control to reduce the spread of COVID-19 disease through the combination of the aspects of the education campaign, quarantine, and treatment of infected individuals. This study intends to apply optimal control theory to minimize the spread disease by some control strategies and minimize the cost of applying controls in order to best combat the spread of COVID-19 disease.\\

\noindent Consider these two epidemic models with controls $\displaystyle u_{1}(t),\ u_{2}(t),\ u_{3}(t) $, death $D(t)$, recovered $R_{1}(t) $ from infected reported and recovered $R_{2}(t) $ from infected unreported, given by the following two  models:
\begin{itemize}
%

\item\textbf{Model 1}:
 \begin{equation}\label{controlmodel2}
		\left\{\begin{array}{lcl}
	\dot S &=& \Lambda -\beta \,S\,(I+\epsilon (1-u_{1}) I_{1}+I_{2}(1-u_{1}))-\mu \,S,\\
	\\
	\dot I &=&\beta \,S\,(I+\epsilon (1-u_{1}) I_{1}+I_{2}(1-u_{1}))-(\mu+\alpha_1 (1+u_{2})+\alpha_2 (1-u_{3}))\,I,    \\
	\\
	\dot I_1&=& \alpha_1 (1+u_{2})\,I - (\mu+d(1-u_{3})+\theta_1(1+u_{3}))\,I_1,\\
	\\
	\dot I_2 &=&  \alpha_2 (1-u_{3})\,I - (\mu+d(1-u_{3})+\theta_2 (1+u_{3}))\,I_2, \\
	\\
	\dot R_{1} &=&\theta_1 (1+u_{3})\,I_1-\mu\,R_{1},\\
	\\
	\dot R_{2} &=&\theta_2 (1+u_{3})\,I_2-\mu\,R_{2},\\
	\\
	\dot D &=&d (1-u_{3})\,(I_1+\,I_2).
	\end{array}\right.
\end{equation}
\item\textbf{Model 2}:
\begin{equation}\label{controlmodel3}
\left\{\begin{array}{lcl}
\dot S &=& \Lambda -\beta \,S\,(I+\epsilon (1-u_{1}) I_{1}+(1-u_{1})I_{2})-\mu \,S,\\
\\
\dot I &=&\beta \,S\,(I+\epsilon (1-u_{1}) I_{1}+(1-u_{1})I_{2})-(\mu+\alpha_1 (1+u_{2}+u_{3})+\alpha_2 (1-u_{2}-u_{3}))\,I,    \\
\\
\dot I_1&=& \alpha_1 (1+u_{2}+u_{3})\,I - (\mu+d(1-u_{3})+\theta_1(1+u_{2}+u_{3}))\,I_1,\\
\\
\dot I_2 &=&  \alpha_2 (1-u_{2}-u_{3})\,I - (\mu+d(1-u_{3})+\theta_2 (1+u_{2}+u_{3}))\,I_2, \\
\\
\dot R_{1} &=&\theta_1 (1+u_{2}+u_{3})\,I_1-\mu\,R_{1},\\
\\
\dot R_{2} &=&\theta_2 (1+u_{2}+u_{3})\,I_2-\mu\,R_{2},\\
\\
\dot D &=&d (1-u_{3})\,(I_1+\,I_2).
\end{array}\right.
\end{equation}
\noindent We consider, in this work, three controls: distancing control $u_{1}(t)$ , case finding control $u_{2}(t)$ , and case holding control $u_{3}(t)$ and $\displaystyle 0\leq\displaystyle u_{1}(t),\ u_{2}(t),\ u_{3}(t)\leq 1$.\\
\begin{itemize}
	\item The distancing control, $u_{1}(t)$: implies the effort of preventing susceptible individuals from becoming infectious individuals. The strategies, such as early detection of infectious individuals, isolation of infectious people, and health campaign and education, are related to $u_{1}(t)$. It reduces the risk of contamination of susceptible by the reported individuals during quarantine and treatment. It also reduces the contamination from unreported, since due to health campaigns and education, individuals who suspect they are infected may apply social distancing and self-quarantine.  
	\item The case-finding control, $u_{2}(t)$: represents the screening of high-risk exposed individuals in the first model and, additionally, the treatment of infected individuals in the second model. It increases the detection of infected individuals and also increases the recovered cases in the second model.
	\item The case holding control, $u_{3}(t)$: refers to the effort required to complete the treatment of infected individuals, such as activities used to ensure the regularity of drug intake until a lasting cure is attained and financial support by the government. 
	It increases the detection of cases in the second model, reduces unreported cases, increases the recovered cases, and reduces death cases in the two models.
\end{itemize}
\end{itemize} 

\subsection{ Modeling the optimal control problem}
In this subsection, we present the optimal control problem we intend to solve. Two strategies are proposed to analyze the spread of the viruses when some controls are applied.\\
%
\noindent Let's set $u(t)=(u_{1}(t),u_{2}(t),u_{3}(t))\in [0,1]^{3}$,  $x(t)=(S(t),I(t),I_{1}(t),I_{2}(t),R_{1}(t),R_{2}(t),D(t))$ and\\ $x_{0}=(S_{0},I_{0},I_{10},I_{20},R_{10},R_{20},D_{0})$.\\
\noindent Our objective functional to be minimized is as follows:
\begin{multline}
J(x,u)= \int_{0}^{T}  [ a_{1}I(t) + a_{2}I_{1}(t) + a_{3}I_{2}(t)+ a_{4}D(t) + \dfrac{a_{5}}{2}u_{1}^{2} (t) +\dfrac{a_{6}}{2}u_{2}^{2} (t) + \dfrac{a_{7}}{2}u_{3}^{2}  (t)]  dt
\label{functionnal}
\end{multline}
We assume that the relative intervention costs are nonlinear and take a quadratic form in the controls. The coefficients, $ a_{i}\ i=1\cdots 7$, are balancing factors according to the size and the importance of the objective functional. 
Thus, we seek optimal controls variables $u^{*}$ and states variables $x^{*}$ such that

\begin{align}
J(x^{*},u{*})&= \min_{u(t)\in [0,\ 1]^{3}} J(x,u)&\nonumber\\
 \textrm{subject to: }&\  x(t) \textrm{ satisfies the DE model \eqref{controlmodel2} or \eqref{controlmodel3}}&\\
 &x(0)=x_{0}&\nonumber
 \label{optprb}
\end{align}
\noindent By minimizing the functional, we want to reduce, at the same time, the infectious asymptomatic, the reported and unreported symptomatic, the death, and the controls. \\
Further, we will propose a constructed functional  $J(x,u)$. For the construction of that functional, refer to section \ref{app}.

\subsection{ Existence of an optimal control solution}
We analyze sufficient conditions for the existence of a solution to the optimal control problem (\eqref{optprb}). Using a result in Fleming and Rishel (\cite{Fleming} ) and Hattaf and Yousfi (\cite{Hattaf}), the existence of the optimal control can be obtained.
\begin{thm}
	There exists an optimal control $u^{*}$ and corresponding state $x^{*}$ to the problem \eqref{optprb}. 
\end{thm}
\begin{proof}
The existence of an optimal control is guaranteed by Corollary 4.1 of  Fleming (\cite{Fleming} ) due the following
 \begin{enumerate}
\item the convexity of the integrand of J with respect to u;
\item  a priori boundedness of the state solutions;
\item  Lipschitz property of the state system with respect to the state variables.
 \end{enumerate}
 Since the functional is continuously differentiable in $t,x, u$, and with the bounded domains of the state $x$ and the control $u$, there exist an optimal control and state $(x^{*},u^{*})$ that minimize the functional \eqref{functionnal}. \\  
\end{proof}

\noindent The following theorem is a consequence of the maximum principle.
\begin{thm}
	Given an  optimal control $u{*}$ and corresponding state $x^{*}$ solutions to the problem \eqref{optprb}, then    there exist adjoint variable  $p$ such that $u{*}, x^{*}$ satisfy the Pontryagin's Maximum Principle.
	\label{thm2}
\end{thm}

\begin{proof}
Since the problem \eqref{optprb} with the model \eqref{controlmodel3} generalize the problem with the model \eqref{controlmodel3}, we only perform a proof for the problem \eqref{optprb} with the model \eqref{controlmodel3}. \\
The theorem is a direct application of Pontryagin's maximum principle (\cite{Pontryagin}). Then the Hamiltonian of the problem \eqref{optprb} is given as follows:
\begin{multline}
H=a_{1} I(t)+a_{2} I_{1}(t)+a_{3} I_{2}(t)+a_{4} D(t)+\frac{a_{5}}{2} u^{2}_{1}(t)+\frac{a_{6}}{2} u^{2}_{2}(t)+\frac{a_{7}}{2} u^{2}_{3}(t)+ \displaystyle \sum_{i=1}^{n} p_{i} g_{i}
\end{multline}
where $g_{i},\ i=1\cdots 7$ denotes the right side of the differential equation of the $i$ \ the state variables, and $p=(p_{1}(t), \ p_{2} (t), \  p_{3} (t), \  p_{4} (t), \  p_{5} (t), \  p_{6} (t), \  p_{7}(t))$ the associated adjoints for the states $x$. Then, we obtain

\begin{multline}
H=a_{1} I(t)+a_{2} I_{1}(t)+a_{3} I_{2}(t)+a_{4} D(t)+\frac{a_{5}}{2} u^{2}_{1}(t)+\frac{a_{6}}{2} u^{2}_{2}(t)+\frac{a_{7}}{2} u^{2}_{3}(t)+\\
p_{7}(t).(d (1-u_{3}(t)) (I_{1}(t)+I_{2}(t)))+p_{3}(t).(\alpha_{1} I(t) (u_{2}(t)+u_{3}(t)+1)-I_{1}(t) (d (1-u_{3}(t))+\mu +\\
\theta_{1} (u_{2}(t)+u_{3}(t)+1)))+p_{4}(t).(\alpha_{2} I(t) (-u_{2}(t)-u_{3}(t)+1)-I_{2}(t) (d (1-u_{3}(t))+\mu +\theta_{2} (u_{2}(t)+u_{3}(t)+1)))+\\
p_{1}(t).(-\beta  S(t) (\epsilon  I_{1}(t) (1-u_{1}(t))+I_{2}(t) (1-u_{1}(t))+I(t))+\Lambda -\mu  S(t))+\\
p_{2}(t).(\beta  S(t) (\epsilon  I_{1}(t) (1-u_{1}(t))+I_{2}(t) (1-u_{1}(t))+I(t))-I(t) (\mu +\alpha_{1} (u_{2}(t)+u_{3}(t)+1))-\\
\alpha_{2} I(t) (-u_{2}(t)-u_{3}(t)+1))+p_{5}(t).(\theta_{1} I_{1}(t) (u_{2}(t)+u_{3}(t)+1)-\mu  R_{1}(t))+\\
p_{6}(t).(\theta_{2} I_{2}(t) (u_{2}(t)+u_{3}(t)+1)-\mu  R_{2}(t))
\end{multline}
\noindent Therefore we can derive the following:
\begin{equation*}
\dot p_{1}= - \dfrac{ \partial H}{\partial S}\,,  \dot p_{2}= - \dfrac{ \partial H}{\partial I}\,, \dot p_{3}= - \dfrac{ \partial H}{\partial I_{1}}\, \ , 
\end{equation*}

\begin{equation*}
\dot p_{4}= - \dfrac{ \partial H}{\partial I_{2}}\, \ , \dot p_{5}= - \dfrac{ \partial H}{\partial R_{1}}\, \ , \dot p_{6}= - \dfrac{ \partial H}{\partial R_{2}}\,\ , \dot p_{7}= - \dfrac{ \partial H}{\partial D}\,
\end{equation*}
with \, $p_{i}(T)=0 , for \ i=1,\ 2,\ 3,\ 4,\ 5,\ 6,\ 7$. evaluated at the optimal controls and the corresponding states, which results in adjoint system of theorem (\ref{thm2}). The Hamiltonian $H$ is minimized with respect to the controls at the optimal controls; therefore, we differentiate $H$ with respect to $u_{1}$ , $u_{2} $, and $u_{3}$ on the set  $\Gamma$ , respectively, thereby obtaining the following optimality conditions:

\begin{align*}
\dfrac{ \partial H}{\partial u_{1}}  =0,\ 
\dfrac{ \partial H}{\partial u_{2}}  =0,\ 
\dfrac{ \partial H}{\partial u_{3}}  =0
\end{align*}

Solving for $u_{1}^{*}$ , $u_{2}^{*} $, and $u_{3}^{*}$, we obtain 

\begin{equation*}
\begin{array}{lll}
u_{1}^ {*}(t) &=& \displaystyle \frac{1}{a_{5}}(-p_{1}(t).(-\beta  S(t) (-\epsilon  I_{1}(t)-I_{2}(t)))-p_{2}(t).(\beta  S(t) (-\epsilon  I_{1}(t)-I_{2}(t))))\,\\
\\
u_{2}^ {*} (t)&=& \displaystyle \frac{1}{a_{6}}(-p_{3}(t).(\alpha_{1} I(t)-\theta_{1} I_{1}(t))-p_{5}(t).(\theta_{1} I_{1}(t))-p_{4}(t).(-\theta_{2} I_{2}(t)-\alpha_{2} I(t))\\
&&\qquad\qquad\qquad\qquad\qquad\qquad\qquad - p_{6}(t).(\theta_{2} I_{2}(t))-p_{2}(t).(\alpha_{2} I(t)-\alpha_{1} I(t)))\,\\
\\
u_{3}^ {*} (t)&=&   \displaystyle \frac{1}{a_{7}}(-p_{7}(t).(-d (I_{1}(t)+I_{2}(t)))-p_{3}(t).(\alpha_{1} I(t)-(\theta_{1}-d) I_{1}(t))\\
&&\qquad\qquad\qquad\qquad\qquad\qquad\qquad -p_{4}(t).(-(\theta_{2}-d) I_{2}(t)-\alpha_{2} I(t))-p_{5}(t).(\theta_{1} I_{1}(t))\\
&&\qquad\qquad\qquad\qquad\qquad -p_{6}(t).(\theta_{2} I_{2}(t))-p_{2}(t).(\alpha_{2} I(t)-\alpha_{1} I(t)))\\
\end{array}
\end{equation*}
This end the proof
\end{proof}

\section{ Numerical simulation of optimal controls}
\label{sec:numsim}
In this section, we show numerical simulation of optimal controls problem. We use ACADO an optimal control solver tool to solve the problems. ACADO solver use direct method that is it starts by discretizing the problem to get a non linear problem (NLP) and at the end solve the NLP problem. The parameters of the model are estimated by fitting data of the cumulative cases of Senegal country. That method of fitting cumulative data cases has been presented in \cite{Magal}, \cite{Balde:2020} and \cite{Balde2:2020}. Details is given in the appendix section. The values of the parameters are:
$\Lambda=0.0914 N/100$, with $N=16743927$ the total population of Senegal. 
$\beta = 1.04756\cdot 10^{-8}$, $\mu = 0.000219$,
$\alpha_{1}= 0.110064$, $\alpha_{2} = (1 - f) \alpha_{1}/f$, with $f = 0.8$;  $\theta_{1} = 1/7,\ \theta_{2} = 1/7,\  d = 0.00194523,\  \epsilon = 0.02$. \\      
The initial conditions are: $t_{0}=0.0166363,\ I_{0} = 190.612,\ I_{10} = 98.3121,\  I_{20} = 24.578,\ S_{0} = N - I_{0},\  R_{10} = 0,\ R_{20} = 0,\  D_{0} = 0$.\\
We constructed and use the following functional $\displaystyle J(x,u)=c k(I(t)+I_{1}(t)+I_{2}(t)+D(t))+\frac{1}{2}(u_{1}^{2}+u_{2}^{2}+u_{3}^{2})$. With $c = 4.9\cdot 10^{-5}$ and $k=1$.\\ 
\noindent We consider different strategies:
\begin{enumerate}
	\item We solve the optimal control problem \eqref{optprb} with the model \eqref{controlmodel2}. 
	\item We solve the optimal control problem \eqref{optprb} with the model \eqref{controlmodel3}.
\end{enumerate} 
The functional are constructed based on the general one \eqref{functionnal}. See Section \ref{app} for more details.   \\
\noindent The figures \ref{fig:states-nc-sc1} show results of the model \eqref{model} where we do not consider controls. 
The figures \ref{fig:states-sc2} and \ref{fig:controls-sc2} show the results related to the strategy 1, while the figures \ref{fig:states-sc3} and \ref{fig:controls-sc3} show the results related to the strategy 2.  

\begin{figure}[H]
	\centering
	\subfloat[\label{fig:states-nc-all-sc1}]{\includegraphics[width=0.45\linewidth]{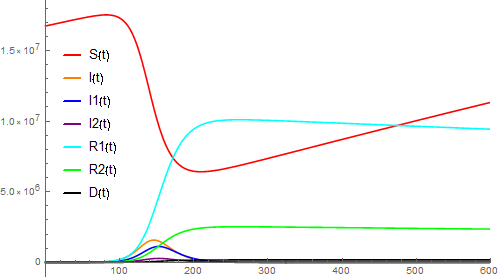}}\qquad\qquad 
	\subfloat[\label{fig:states-nc-all-sc1-1-e6}]{\includegraphics[width=0.45\linewidth]{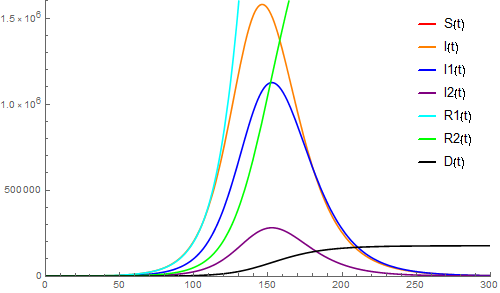}}\qquad\qquad
	\subfloat[\label{fig:states-nc-all-sc1-5-e5}]{\includegraphics[width=0.45\linewidth]{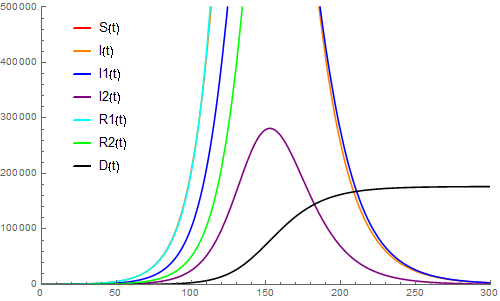}}\qquad\qquad
	\subfloat[\label{fig:states-nc-all-sc1-2-e5}]{\includegraphics[width=0.45\linewidth]{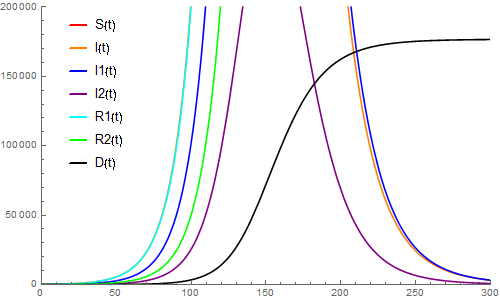}}\qquad\qquad 
	\caption{Plot of the differential equation \eqref{model}. There is no controls strategies.}
	\label{fig:states-nc-sc1}	
\end{figure}

\begin{figure}[H]
	\centering
	\subfloat[\label{fig:u1-sc2}]{\includegraphics[width=0.45\linewidth]{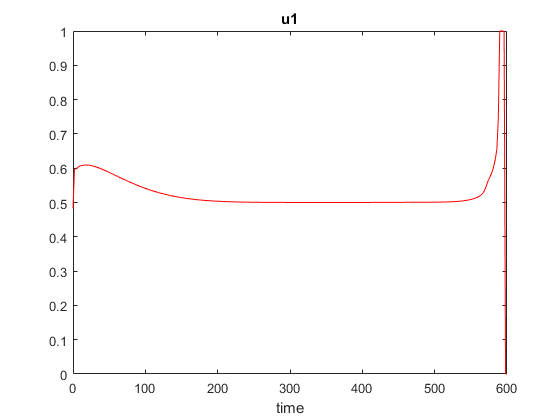}}\qquad\qquad 
	\subfloat[\label{fig:u2-sc2}]{\includegraphics[width=0.45\linewidth]{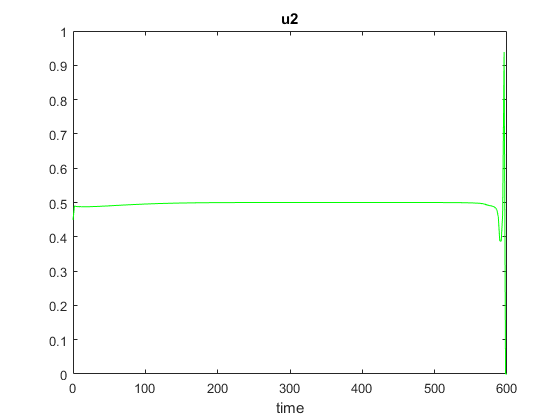}}\qquad\qquad
	\subfloat[\label{fig:u3-sc2}]{\includegraphics[width=0.45\linewidth]{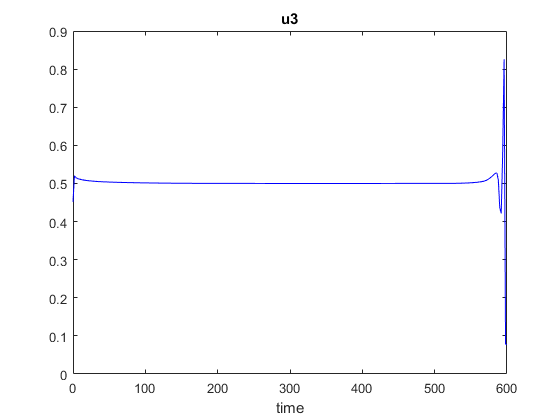}}\qquad\qquad 
	\subfloat[\label{fig:u1u2u3-sc2}]{\includegraphics[width=0.45\linewidth]{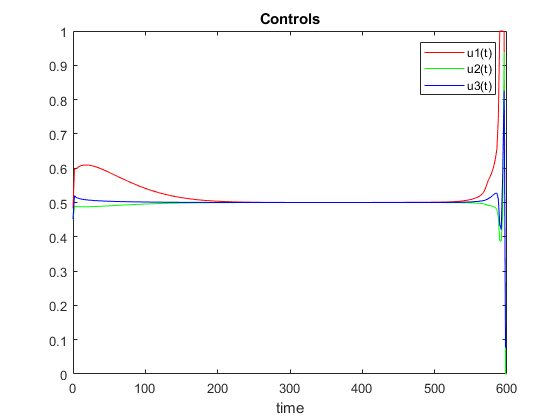}}\qquad\qquad
	\caption{Plot of controls of the first strategy. }
	\label{fig:controls-sc2}	
\end{figure}
\begin{figure}[H]
	\centering
	\subfloat[\label{fig:states-all-sc2}]{\includegraphics[width=0.45\linewidth]{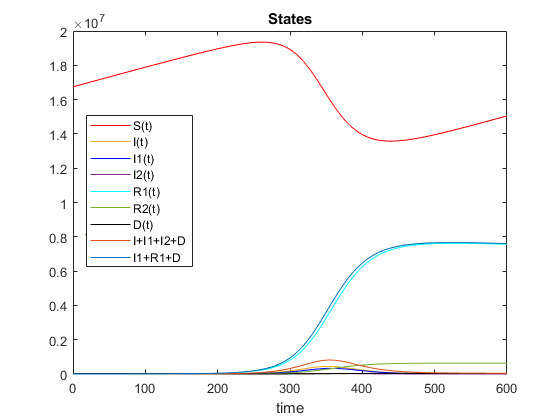}}\qquad\qquad 
	\subfloat[\label{fig:states-all-sc2-9-e5}]{\includegraphics[width=0.45\linewidth]{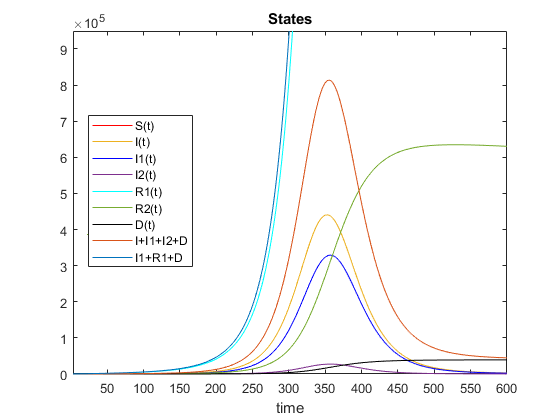}}\qquad\qquad 
	\subfloat[\label{fig:states-all-sc2-3-e5}]{\includegraphics[width=0.45\linewidth]{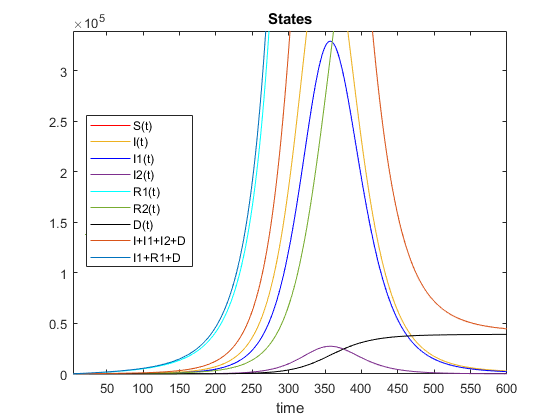}}\qquad\qquad
	\subfloat[\label{fig:states-all-sc2-5-e4}]{\includegraphics[width=0.45\linewidth]{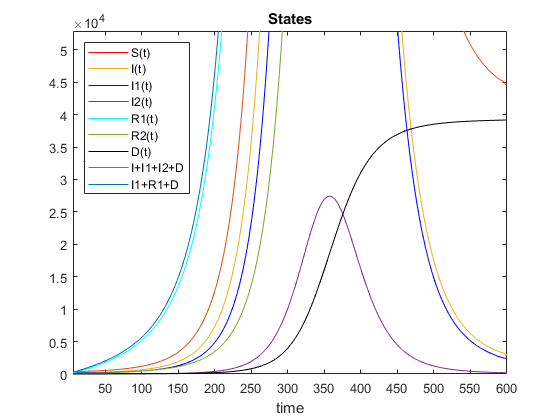}}\qquad\qquad
	\caption{Plot of states of the first strategy. }
	\label{fig:states-sc2}	
\end{figure}

\begin{figure}[H]
	\centering
	\subfloat[\label{fig:u1-sc3}]{\includegraphics[width=0.45\linewidth]{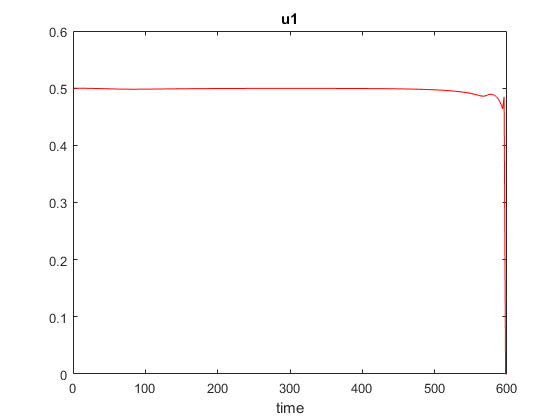}}\qquad\qquad 
	\subfloat[\label{fig:u2-sc3}]{\includegraphics[width=0.45\linewidth]{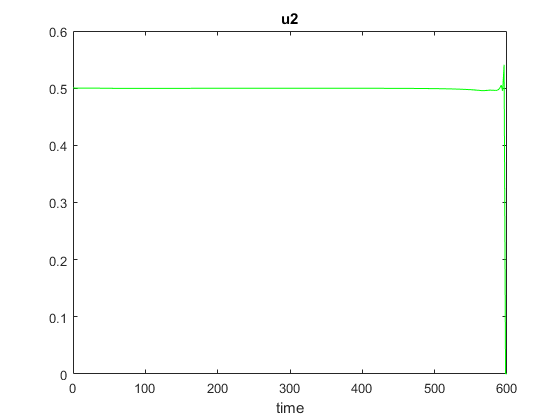}}\qquad\qquad
	\subfloat[\label{fig:u3-sc3}]{\includegraphics[width=0.45\linewidth]{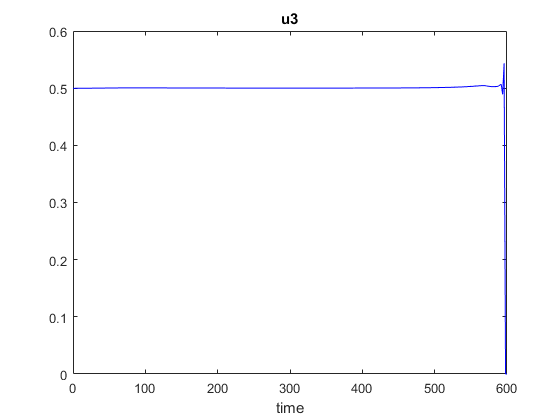}}\qquad\qquad 
	\subfloat[\label{fig:u1u2u3-sc3}]{\includegraphics[width=0.45\linewidth]{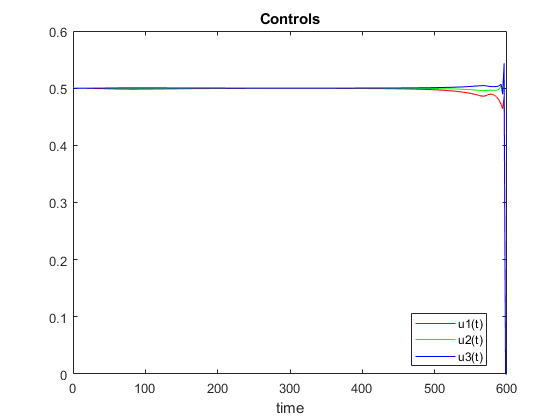}}\qquad\qquad
	\caption{Plot of controls of the second strategy. }
	\label{fig:controls-sc3}	
\end{figure}
\begin{figure}[H]
	\centering
	\subfloat[\label{fig:states-all-sc3}]{\includegraphics[width=0.45\linewidth]{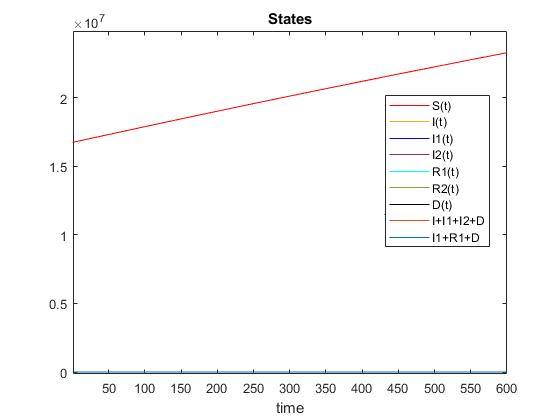}}\qquad\qquad 
	\subfloat[\label{fig:states-all-sc3-1-e3}]{\includegraphics[width=0.45\linewidth]{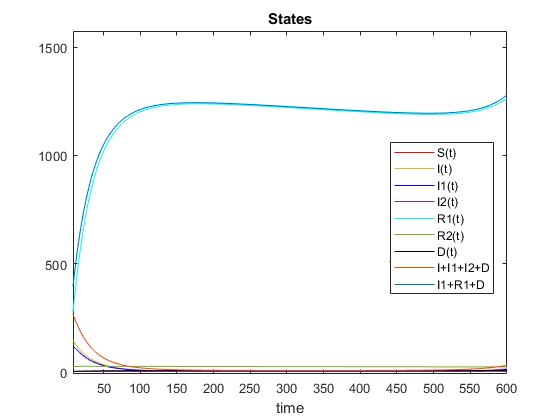}}\qquad\qquad
	\subfloat[\label{fig:states-all-sc3-3-e2}]{\includegraphics[width=0.45\linewidth]{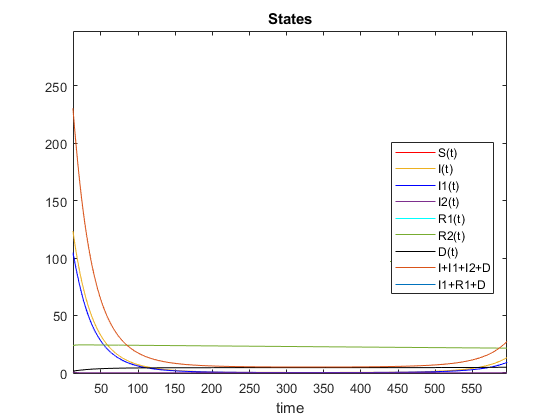}}\qquad\qquad 
	\caption{Plot of states of the second strategy. }
	\label{fig:states-sc3}	
\end{figure}

\begin{figure}[H]
	\centering
	\subfloat[\label{fig:CompInf}]{\includegraphics[width=0.45\linewidth]{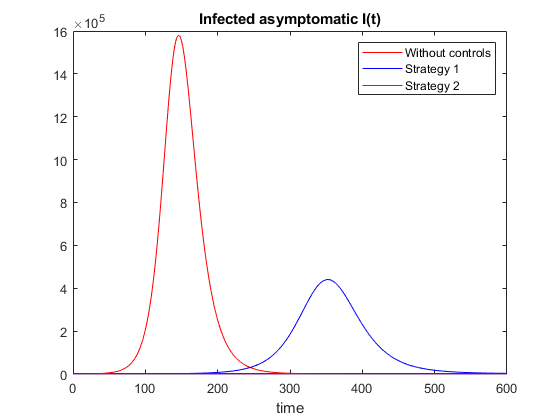}}\qquad\qquad 
	\subfloat[\label{fig:CompInfReport}]{\includegraphics[width=0.45\linewidth]{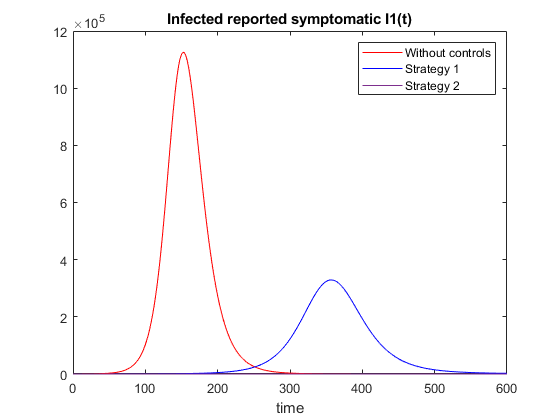}}\qquad\qquad 
	\subfloat[\label{fig:CompInfUnreport}]{\includegraphics[width=0.45\linewidth]{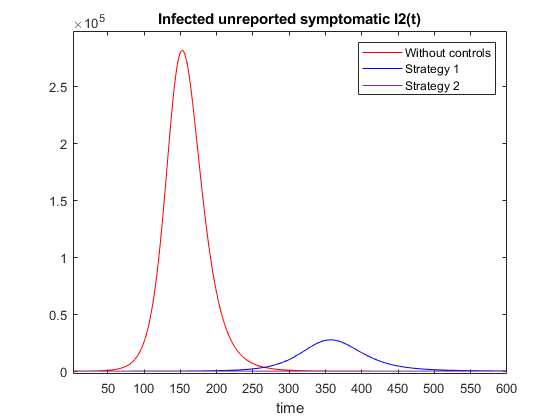}}\qquad\qquad
	\subfloat[\label{fig:CompDeath}]{\includegraphics[width=0.45\linewidth]{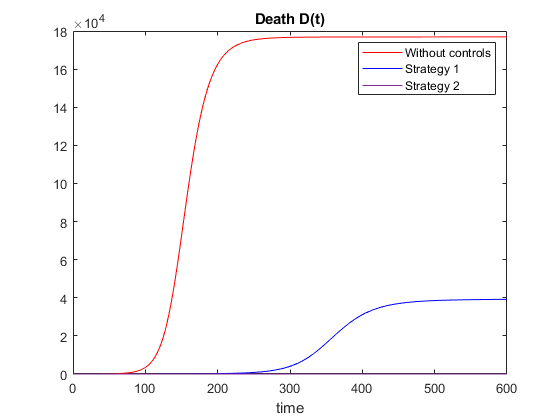}}\qquad\qquad 
	\subfloat[\label{fig:CompInfallandDeath}]{\includegraphics[width=0.45\linewidth]{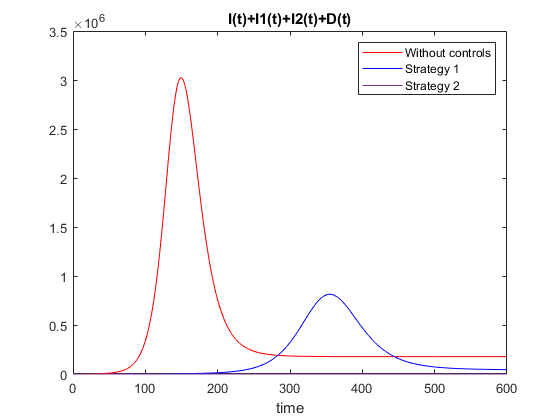}}\qquad\qquad
	\subfloat[\label{fig:CompCumul}]{\includegraphics[width=0.45\linewidth]{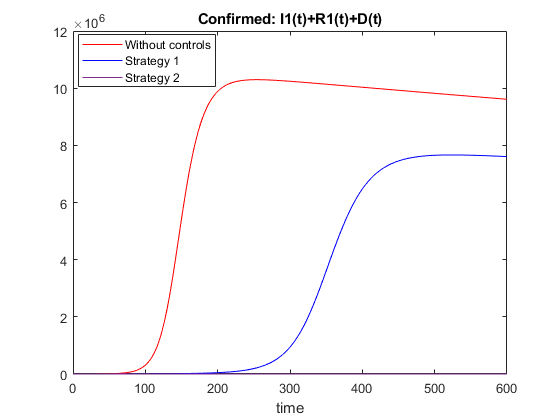}}\qquad\qquad
	\caption{Comparative plot of states with and without controls.}
	\label{fig:Comparative}	
	
\end{figure}

\begin{figure}[H]
	\centering
	\subfloat[\label{fig:CompInf-2-e3}]{\includegraphics[width=0.45\linewidth]{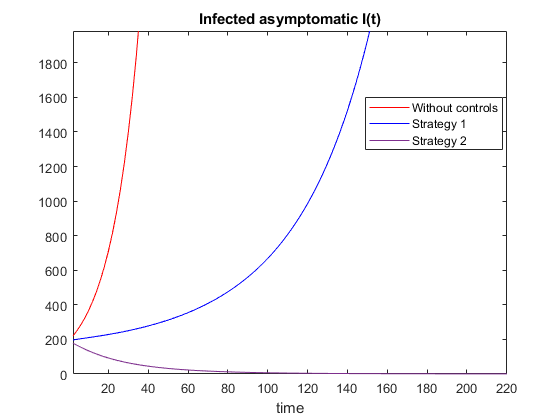}}\qquad\qquad 
	\subfloat[\label{fig:CompInfReport-1-e3}]{\includegraphics[width=0.45\linewidth]{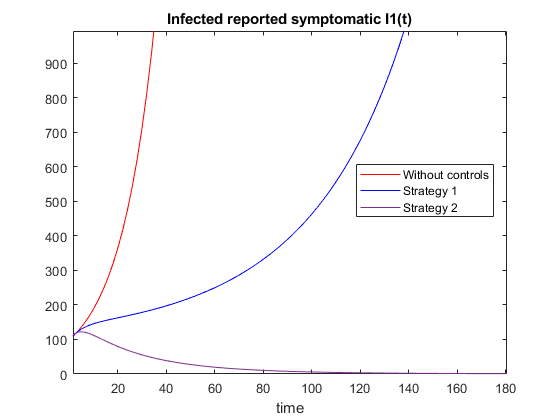}}\qquad\qquad 
	\subfloat[\label{fig:CompInfUnreport-9-e2}]{\includegraphics[width=0.45\linewidth]{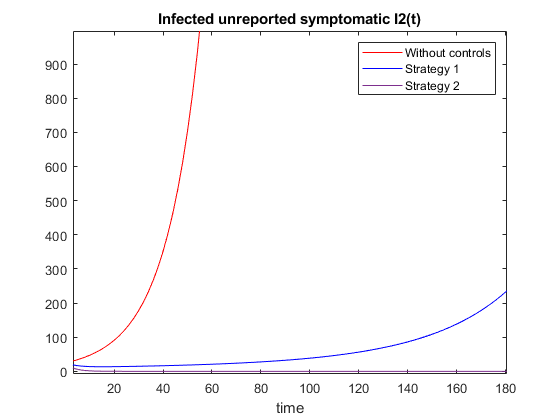}}\qquad\qquad
	\subfloat[\label{fig:CompDeath-9-e1}]{\includegraphics[width=0.45\linewidth]{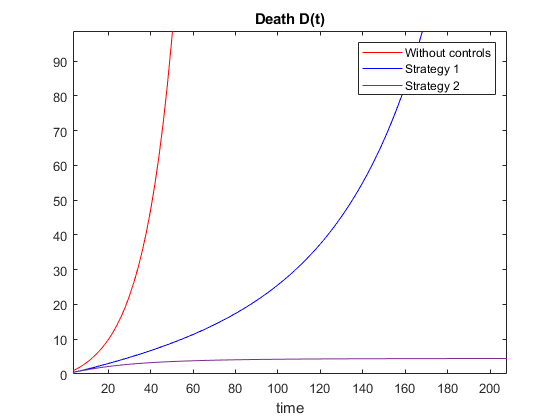}}\qquad\qquad 
	\subfloat[\label{fig:CompInfallandDeath-2-e3}]{\includegraphics[width=0.45\linewidth]{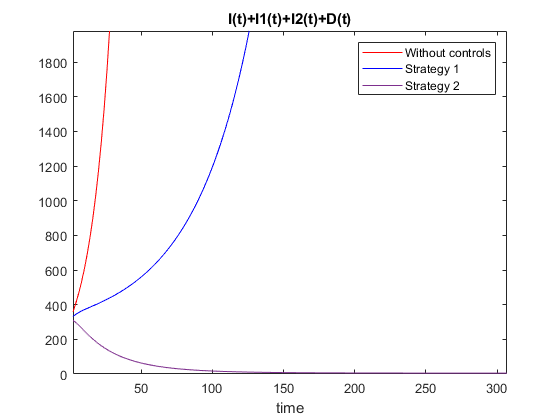}}\qquad\qquad
	\subfloat[\label{fig:CompCumul-2-e4}]{\includegraphics[width=0.45\linewidth]{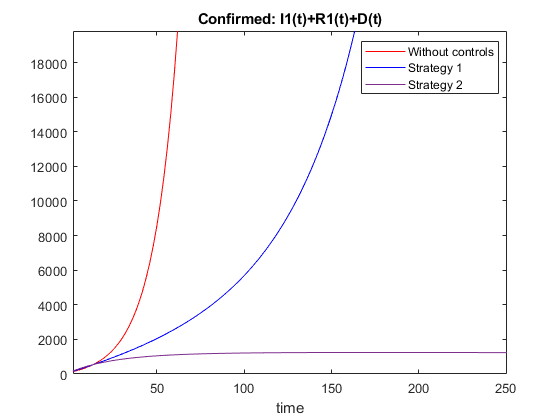}}\qquad\qquad
	\caption{Zoom of comparative plot of states with and without controls showing the effect of optimal controls strategies on the infected and death population.}
	\label{fig:Comparativezoom}	
\end{figure}

\section{Discussion}
\label{sec:discu}
\noindent The functional we minimize has two parts. The first part is composed of infected and death terms, while the second part is composed of the control terms. The size scale of these two parts is very different. Hence the choice of the coefficients $a_{1},\cdots,a_{7}$, as balance, has a great influence on the result. \\
We construct a function with an economic sens by using some data. More details on the construction of that function can be found in the appendix. \\

\noindent We see in the figures \ref{fig:states-nc-sc1}, \ref{fig:states-sc2} and \ref{fig:states-sc3}, that for all the two strategies the infected asymptomatic individuals $I$, the infected reported individuals $I_{1}$ and the infected unreported individuals $I_{2}$ are reduced compared to the case without controls. Also, the death cases are reduced with controls in comparison to the situation without controls.\\

\noindent The second strategy is better than the first one. Indeed, in the second strategy, there is no epidemic.\\

\noindent The results in the figures \ref{fig:Comparative} and \ref{fig:Comparativezoom} show a difference in the number of infected (Reported/ Unreported) individuals without controls compared to the number with optimal strategies. Due to the control strategies $1$ and/or $2$, the number of infected individuals (Reported/ Unreported) decreases and reaches the turning point of the asymptomatic infectious cases later than without controls. At the same time, optimal strategies reduced the maximal number of infected (Reported/ Unreported) people compared to the case without controls.
In other words, the maximal value of the peak decreases, and the time of the peak is postponed by applying controls. 
\section{Material and Methods}
\label{app}
\subsection{Estimation of parameters}
\label{estima}
The estimation of the parameters of the model \eqref{model} is done by using techniques in \cite{Magal}, \cite{Balde:2020} and \cite{Balde2:2020}. We fit the cumulative data with an exponential function $TNI(t)=b\exp(c t)-a$. In addition, we assume that the cumulative function can be given in integral form as $TNI(t)=\displaystyle\alpha_{1}\int_{t_{0}}^{t}I(s)ds+TNI_{0}$.\\ Then $TNI(t_{0})=TNI_{0}=b\exp(c t_{0})-a$. Thus, we obtain $\displaystyle t_{0}=\displaystyle\frac{\ln(TNI_{0}+a)-\ln(b)}{c}$.\\
Also, we have:
\begin{equation}
I(t)=\dot{TNI}(t)=bc\exp(c t).
\end{equation} 
Then $\displaystyle I(t_{0})=\frac{bc}{\alpha_{1}}\exp(c t_{0})=\frac{c}{\alpha_{1}}(TNI_{0}+a)=I_{0}$ and $\displaystyle\frac{I(t)}{I(t_{0})}=\exp(c(t-t_{0}))$. 
Hence, we obtain
\begin{equation}
I(t)=I(t_{0})\exp(c(t-t_{0})),
\end{equation}
then $\dot{I}(t)=cI(t)$ and $\dot{I}(t_{0})=cI(t_{0})$.\\ 
Let's set $\delta_{1}$ and $\delta_{2}$ such that $I_{1}=\delta_{1} I$ and $I_{2}=\delta_{2} I$. 
Then replacing in the second an third equation of the following system: 
\begin{equation}
\label{modelinf}
	\left\{\begin{array}{lcl}
	\dot I &=&\beta \,S\,(I+\epsilon I_{1}+I_{2})-(\mu+\alpha_1+\alpha_2)\,I,    \\
	\dot I_1&=& \alpha_1 \,I - (\mu+d+\theta_1)\,I_1,\\
	\dot I_2 &=&  \alpha_2 \,I - (\mu+d+\theta_2)\,I_2, \\
	\end{array}\right.
\end{equation}
 we obtain 
 \begin{align}
 	\delta_{1}=\frac{\alpha_1}{c+\mu+d+\theta_1}=\frac{I_{10}}{I_{0}}&\\
 	\delta_{2}=\frac{\alpha_2}{c+\mu+d+\theta_2}=\frac{I_{20}}{I_{0}}&
 \label{delta}
 \end{align}
Then introducing \eqref{delta} in the first equation of \eqref{modelinf}, we obtain:
\begin{equation*}
	c+\mu+\alpha_1+\alpha_2=\beta \,S_{0}\,(1+\epsilon \delta_{1}+\delta_{2})
\end{equation*}  
Hence 
 \begin{align}
 \label{beta}
\beta=\frac{c+\mu+\alpha_1+\alpha_2}{S_{0}\,(1+\epsilon \delta_{1}+\delta_{2})}
\end{align}
Replacing \eqref{delta} in \eqref{beta}, we obtain:
 \begin{align}
\beta=\frac{(c+\mu+\alpha_1+\alpha_2)(c+\mu+d+\theta_1)(c+\mu+d+\theta_2)}{S_{0}\,((c+\mu+d+\theta_1)(c+\mu+d+\theta_2)+\epsilon \alpha_{1}(c+\mu+d+\theta_1)+\alpha_{2}(c+\mu+d+\theta_2))}
\end{align}
To estimate the death rate, we have:
 \begin{align}
D_{1}(t)=&\int_{t_{0}}^{t} d I_{1}(s)ds=\int_{t_{0}}^{t} d \delta_{1}I(s)ds\\
=&\frac{d}{c+\mu+d+\theta_1}(b \exp(c t)-a-TNI_{0}).
\end{align}
With $D_{1}$, the death from reported cases.\\

\noindent We consider that $80\%$ of cases can be detected. Then $f=0.8$ and $\alpha_{2}=\displaystyle\frac{1-f}{f}\alpha_{1}$, with $\alpha_{1}$ estimated above. We set the infectious period to medical values $1/7$ for all infected reported and unreported. The pandemic death rate $d$ is estimated by using reported death data. We consider the same value for death from unreported cases.  \\
For the birth rate, we use $32.9\% $ of year $2018$, from \url{https://fr.wikipedia.org/wiki/Démographie_du_Sénégal}. Then the recruitment is $\Lambda=32.9\% N/365$ by day. The death rate is $7.9\% $ by year at $2018$.    
\subsection{Construction of the  
functional}
\label{func}
In order to have a functional with economic sens, we consider what follows:
\begin{itemize}
	\item We lost money when people are infected or death. Let's note that cost $a$ by individual and by day, associated to $I$, $I_{1}$, $I_{2}$ and , $D$. Then $a=a_{1}=a_{2}=a_{3}=a_{4}$.
	\item We lost money when we perform test on susceptible individuals. Let's note that cost $a_{6}$ by day, associated to controls $u_{2}$.  
	\item We spend money to provide treatment. Let's note that cost $a_{7}$ by day, associated to the control $u_{3}$.
	\item We spend money to carry out health campaigns and education. Let's note that cost $a_{5}$ by day, associated to the control $u_{1}$. 
\end{itemize}
Then the functional \eqref{functionnal} become: $\displaystyle J(u_{1},u_{2},u_{3})=a(I(t)+I_{1}(t)+I_{2}(t)+D(t))+a_{5}\frac{u_{1}^{2}}{2}+a_{6} \frac{u_{2}^{2}}{2}+a_{7} \frac{u_{3}^{2}}{2}$. We consider that $J(u_{1},u_{2},u_{3})\leq C$, with $C$ a maximal expense. Then we can rewrite the functional with proportion coefficients $\displaystyle a^{'}=\displaystyle \frac{a}{C}, \ a_{5}^{'}=\displaystyle \frac{a_{5}}{C}, \ a_{6}^{'}=\displaystyle \frac{a_{6}}{C}, \ a_{7}^{'}=\displaystyle \frac{a_{7}}{C}$:
$$
J(u_{1},u_{2},u_{3})=a^{'}(I(t)+I_{1}(t)+I_{2}(t)+D(t))+a_{5}^{'} \frac{u_{1}^{2}}{2}+a_{6}^{'} \frac{u_{2}^{2}}{2}+a_{7}^{'}\frac{u_{2}^{2}}{2}\leq 1.
$$ 
To characterize the loss of money due to infected and death individuals, we use the GDP per capita. The Gross domestic product (GDP) per capita is an indicator of the level of economic activity. It is the value of GDP divided by the number of inhabitants of a country. This indicator is sometimes used to roughly measure per capita income. See \url{https://fr.wikipedia.org/wiki/Produit_intérieur_brut_par_habitant}.\\
Now considering the $2018$ GDP per capita and per day of the Senegal country evaluated to $2456.334$ $ FCFA$ calculated from $1522\ \$ $ per capita, per year. The data come from \url{https://www.populationdata.net/pays/senegal/}. We set $a=2456.334$. \\
A COVID-19 test in Senegal country is evaluated to $50000\ FCFA$ by individual. If we fix a number of test to perform at $1000$ by day, then we have $a_{6}=50000000$. We see that $a=a_{6}\cdot 4.91267\cdot 10^{-5}$. We set $c=4.91267\cdot 10^{-5}$. We choose to fix $a_{5},\ a_{7}$ to the same value of $a_{6}$. Then the functional becomes:
   $$
   J(u_{1},u_{2},u_{3})=ca_{6}^{'}(I(t)+I_{1}(t)+I_{2}(t)+D(t))+a_{5}^{'}\frac{u_{1}^{2}}{2}+a_{6}^{'}\frac{u_{2}^{2}}{2}+a_{7}^{'}\frac{u_{2}^{2}}{2} \leq 1.
   $$ 
Thus considering that the costs are proportional to there respective control, we write:
\begin{align*}
J(u_{1},u_{2},u_{3})=&c\frac{u_{2}}{2}(I(t)+I_{1}(t)+I_{2}(t)+D(t))+\frac{u_{1}^{2}}{2}+\frac{u_{2}^{2}}{2}+\frac{u_{3}^{2}}{2}\\
\leq& \frac{c}{2}(I(t)+I_{1}(t)+I_{2}(t)+D(t))+\frac{u_{1}^{2}}{2}+\frac{u_{2}^{2}}{2}+\frac{u_{3}^{2}}{2}.
\end{align*}
        
\noindent Finally to generalize the functional we obtain:
$$
J(u_{1},u_{2},u_{3})=ck(I(t)+I_{1}(t)+I_{2}(t)+D(t))+\frac{u_{1}^{2}}{2}+\frac{u_{2}^{2}}{2}+\frac{u_{3}^{2}}{2}.
$$     
\section{Conclusion and perspectives}
\label{cp}
In this work, we solve optimal control problems. A new epidemic model, with confirmed contamination, has been presented. We use distancing, case finding, and case holding controls to reduce the spread of the epidemic. We mathematically analyze the model and estimates the parameters used to solve the optimal control problems. In this particular research, the trend of population dynamics is important. It can be easily seen that by increasing educational campaigns, disease tests, and financial support ensure drug for infected individuals, we can successfully decrease the number of infected and death. \\
In further work, we intend to use models with additional compartments as quarantine and treatment.  

%


\end{document}